\documentclass[12pt,a4paper]{article}%
\usepackage{amsmath}
\usepackage{mathptmx}
\DeclareMathOperator{\Tr}{Tr}
\DeclareMathAlphabet{\mathpzc}{OT1}{pzc}{m}{it}
\usepackage{amsfonts}
\usepackage{amssymb}
\usepackage{graphicx}
\usepackage{tikz, tkz-euclide}%
\providecommand{\U}[1]{\protect\rule{.1in}{.1in}}
\newtheorem{theorem}{Theorem}

\newtheorem{definition}[theorem]{Definition}

\newtheorem{remark}[theorem]{Remark}

\newenvironment{proof}[1][Proof]{\noindent\textbf{#1.} }{\ \rule{0.5em}{0.5em}}
\begin{document}

\title{Uniform bound of the entanglement for the ground state of the one-dimensional quantum Ising
model  with non-homogeneous transverse field}
\author{Massimo Campanino, \\
University of Bologna \\
massimo.campanino@unibo.it}
\maketitle

\begin{abstract}
We consider the ground state of the one-dimensional quantum Ising model with transverse field
$h_x$ in one dimension depending on the site $x \in \mathbb Z$ in a finite volume 
$\Lambda_{m}:=\{-m,-m+1,\ldots,m+L\}\ $. We make suitable assumptions on the regions where the field is small and prove that if the field is sufficiently large on the complementary set, then
 the entanglement of the
interval $\Lambda_{0}:=\left\{  0,..,L\right\}  $ relative to its complement
$\Lambda_{m}\backslash\Lambda_{0}$ is bounded  uniformly in $m$ and $L$.
The result applies in particular to periodic transverse fields. The bound
is established by means of a suitable cluster expansion.

\end{abstract}

\bigskip

\begin{description}
\item[AMS\ subject classification:] {\small 60K35, 82B10, 82B31. }

\item[Keywords and phrases:] {\small Quantum Ising model, Entanglement,
Spin-flip processes, Gibbs random fields, Cluster expansion. }
\end{description}

\bigskip

\section{Introduction}

The entropy of entanglement of the ground state of quantum spin systems is a quantity of interest in statistical mechanics.
In particular one is interested in its behaviour as a function of the volume. 

In the case of the one-dimensional quantum
Ising model in \cite{GOS} for values of the field above some percolation threshold a bound that goes 
as the logarithm of the volume was obtained with a stochastic representation in terms of
a FK random cluster model (see \cite{DLP}, \cite{CKP},  \cite{AKN}). In \cite{CG2} a uniform bound was proved
for sufficiently large transversal field by using a cluster expansion for spin-flip processes that
was developed in \cite{CG}. In \cite{GOS2} uniform bound
was proved above the critical point of the model extending the methods of  \cite{GOS}.

It is natural to consider the non-homogeneous case i. e. when the transverse field
$h_x$ depends on the site. It is easy to see that previous results extend to the
case when the $h_x$'s are all above the value obtained for the homogeneous case.
In \cite{GOS2} the question was raised to prove uniform bound of the entanglement
 entropy if the transverse field is the realization of a sequence of i.~ i.~d. 
One can consider the case when the  $h_x$'s  are the realization of a sequence of
i. i. d. random variables with a density which is positive at $0$. The conjecture is that
a uniform bound should hold at high disorder, i. e. if $h_x= \lambda \xi_x$, with  $\xi_x$ sequence of i.~ i.~ d.
random variables with density positive at $0$, when $\lambda$ is sufficiently large.
For this model in \cite{CKP}  decay of correlation functions of the ground state was proved at high
disorder. 

Here we will prove uniform bound in the non-homogeneous case
in cases where  the non-homogeneous transverse field can
be arbitrarily small (but bounded from below) at single sites, provided that the clusters (i.~e. intervals since we are in one dimension) of these sites
have bounded size and the field is sufficiently large on the complementary set. This does not include
the case with i.~ i.~ d. transverse field that was mentioned above, but applies for example to
periodic transverse fields provided that one of their values is sufficiently large.

The ground state of the quantum Ising model with a transverse magnetic field
can be represented as a classical Ising model with one added continuous
dimension \cite{DLP}. In turn this classical Ising model can be represented
via a suitable FK random cluster model \cite{FK}, \cite{CKP}, \cite{AKN}. This
last representation has been used for example in \cite{GOS} \cite{GOS2}
 to study the
entanglement of the ground state in the supercritical regime.

We consider the ground state of the quantum Ising model with non-homogeneous transverse field
$h_x$ in one dimension in a finite volume
\begin{equation}
\Lambda{_{m}:=\{-m,-m+1,\ldots,m+L\}\ .}%
\end{equation}
If $\mathcal{H}_{m}:=\mathcal{H}_{\Lambda_{m}}$ is the Hilbert space for the
quantum system defined on $\Lambda_{m},$ considering the representation of
$\mathcal{H}_{m}$ as $\mathcal{H}_{m,L}\otimes\mathcal{H}_{L},$ with
$\mathcal{H}_{L}:=\mathcal{H}_{\Lambda_{0}}$ and $\mathcal{H}_{m,L}%
:=\mathcal{H}_{\Lambda_{m}\backslash\Lambda_{0}},$ let $\rho_{m}^{L}$ be the
trace over $\mathcal{H}_{m,L}$ of the density operator associated to the
ground state of the system.  We consider the entanglement entropy of the interval $[0, L]$ relative to its complement
$\Lambda_m \backslash [0, L] $.

 We allow the non-homogeneous field to be arbitrarily small 
(but bounded from below) in some region $A \subset \mathbb{Z}$  that is sufficiently well spaced, in the sense
that the size of clusters (i.~e. intervals) of $A$ is bounded. Under this conditions
we prove that there is a constant $C_2$ such that if the transverse field is larger
than $C_2$ on the complementary
of $A$, then the entanglement entropy is uniformly bounded.

In section \ref{System} we recall the definition of the quantum Ising model
with transverse field on $\mathbb{Z}$ and of entanglement entropy.

In section \ref{spin-flip} we recall  the spin-flip representation of the ground state. 

In section \ref{Result} we state and prove the main result.

\section{The model} \label{System}

We consider the Hilbert space $\mathcal{H}:=l^{2}\left(  \left\{
-1,1\right\}  ,\mathbb{C}\right)  $ which is isomorphic to $\mathbb{C}^{2}.$
The algebra $\mathcal{U}:=M\left(  2,{\mathbb{C}}\right)  $ of bounded linear
operators on $\mathcal{H}$ is then generated by the Pauli matrices
$\sigma^{\left(  i\right)  },i=1,2,3$ and by the identity $I.$ In particular,
unless differently specified, in the following we will always consider the
representation of $\mathcal{U}$ with respect to which $\sigma^{\left(
3\right)  }$ is diagonal, i.e.
\begin{equation}
\sigma^{\left(  3\right)  }=\left(
\begin{array}
[c]{cc}%
1 & 0\\
0 & -1
\end{array}
\right)
\end{equation}
and
\begin{equation}
\sigma^{\left(  1\right)  }=\left(
\begin{array}
[c]{cc}%
0 & 1\\
1 & 0
\end{array}
\right)  \ .
\end{equation}

Let $\Lambda$ be a finite connected subset of $\mathbb{Z}$ and set
$\mathcal{H}_{\Lambda}:=\bigotimes_{x\in\Lambda}\mathcal{H}_{x}$ where, for
any $x\in\Lambda,\mathcal{H}_{x}$ is a copy of $\mathcal{H}$ at $x.$ The
finite volume Hamiltonian of the ferromagnetic quantum Ising model with
transverse field is the linear operator on $\mathcal{H}_{\Lambda}$%
\begin{equation}
H_{\Lambda}\left(  J,h\right)  :=-\frac{1}{2}J\sum_{ \langle  x,y \rangle}  \sigma_{x}^{\left(  3\right)
}\sigma_{y}^{\left(  3\right)  }-\sum_{x\in\Lambda} h_{x}\sigma_{x}^{\left(
1\right)  }\ , \label{HtfI}%
\end{equation}
with $h_x>0$ and $J\geq0$, where  $ \langle  x,y \rangle$ indicates that $x$ and $y$ are nearest neighbours. 

Let $H_m$ denote $H_{\Lambda_m}\left(  J,h\right) $ $| and \psi_m \rangle$ denote the ground state of the operator $H_m$
(i.~e.  the eigenvector corresponding to the lowest eigenvalue). If we write
\begin{equation}
\rho_m ( \beta) = \frac {e^{- \beta H_m}} {\Tr \left( e^{- \beta H_m} \right)}.
\end{equation}
 .we have
\begin{equation}
\rho_m  = \lim _{\beta \to \infty}   \rho_m ( \beta) = | \psi_m \rangle \langle \psi_m |.
\end{equation}

Let $\rho^L_m$ be the trace of $\rho_m$ relatively to the Hilbert space corresponding to $\Lambda_m \backslash [0, L] $.

Let us denote by $\lambda^{\downarrow}_j(\rho^L_m)$ the $j$-th eigenvalue of $\rho^L_m$
in decreasing order.

\begin{definition} The entanglement entropy of the interval $[0, L]$ relative to its complement
$\Lambda_m \backslash [0, L] $ is given by
\begin{equation}
S( \rho^L_m) = - \Tr ( \rho^L_m \log_2 \rho^L_m) =- \sum_{j=1}^{2^{L+1}} \lambda^{\downarrow}_j(\rho^L_m) \log_2 (\lambda^{\downarrow}_j(\rho^L_m) ).
\end{equation}
\end{definition}

\section{Spin-flip process representation of the system} \label{spin-flip}

The ground state can be studied by introducing an extra continuous dimension (see \cite{DLP}).
Given a closed interval $I$, let $\mu_{h, I}$  be the probability measure on spin configurations on $I$
with values on $\{-1, 1 \}$ defined starting from a Poisson point process on $I$ with intensity $h$:
 the spins are assumed to be constant on each interval of the complementary of the points of the process and to change value on two consecutive intervals, moreover $\mu_{h, I}$ is assumed to be symmetrical with respect to spin inversion. The values of the spins at the flipping points is left undefined,
as it will be irrelevant.  For $x in \mathbb{Z}$ and $t \in I$  $\sigma_{x}\left(  t\right)$ is 
the value of the spin configuration $\sigma_{x}$ at the point $t$.

Given an interval $\Lambda \subset \mathbb{Z}$ and a real $\beta > 0$, one defines  a finite
volume Gibbs measure on the configurations in $\Lambda \times [- \frac {\beta} {2},  \frac {\beta} {2}]$  .
This is the probability measure with density

\begin{align}
&  Z_{\Lambda}^{-1}  \exp\left[
J\sum_{\langle x,y\rangle }\int
_{-\frac{\beta}{2}}^{\frac{\beta}{2}}\sigma_{x}\left(  t\right)
\sigma_{y}\left(  t\right)  \text{d}t
  \right]  \label{Gibbs1} 
\end{align}
 
with respect to the probability measure

\begin{align}
& \mathcal{P}_{\Lambda, \beta} =\bigotimes_{x \in \Lambda} \mu_{h_x,    [- \frac {\beta} {2},  \frac {\beta} {2}]},
\end{align}

where $\langle x,y\rangle$ indicates that $x$ and $y$ are nearest neighbours and $  Z_{\Lambda}^{-1} $ denotes the normalizing constant.

Let $\Lambda=[x_1, x_2]$. The limiting Gibbs distribution probability in the volume $\Lambda \times [- \frac {\beta} {2},  \frac {\beta} {2}]$ 
as $\beta \to \infty$ of the configuration $\epsilon_{x_1}, \ldots, \epsilon_{x_2}$ on $\Lambda \times \{ 0 \}$ is given by
\begin{align} \label{GSdistribution}
&r(\epsilon_{x_1}, \ldots, \epsilon_{x_2})= \\
&=| \langle \eta_{x_1} \otimes \ldots \otimes \eta_{x_2} \psi \rangle |^2,
\end{align}
where $\psi$ is the ground state and $\eta_{x_1} \otimes \ldots \otimes \eta_{x_2}$,
is the basis with the $\eta$'s equal either to $(1, 0)$ or  to $(0, 1)$.
According to the definition that we gave, we have imposed free boundary condition, but any other boundary conditions give the same result.  

\section{Transfer matrix}

\begin{definition}

Given a subset $A$ of $\mathbb{Z}$, we define cluster of $A$  a maximal interval contained in $A$.
    
\end{definition}

Consider a cluster $I$ of points in $A$  and for some $s >0$ let
$M=[0,s]$

We denote by $g_s( \sigma, \sigma')$ the result
of the integral  over the Poisson point
processes in  $I \times M$, where $\sigma$ and $\sigma'$ are the spin configurations at $I \times \{0 \}$ and $I \times \{s\}$.
We can write 

\begin{equation} \label{operator}
g_s( \sigma, \sigma') =
\int \exp \left( -J \int_a ^b\sum_{j=0}^{l-2}\sigma_j(t) \sigma_{j+1}(t)
\right)  \text{d}\bigotimes_{i=0}^{l-1} \mu_{h_i, [a,b]},
\end{equation}
 where the trajectories have prescribed initial and final values
 $\sigma$ and $\sigma'$ and jumps prescribed by the Poisson point processes. The matrix $T$ with entries given
 $g_s( \sigma, \sigma') $, called \textit{transfer matrix}, has positive entries and by Perron-Frobenius (\cite{F}) has a positive
 largest eigenvalue $\lambda$ and the remainder of the spectrum in a circle of strictly smaller radius. It is easy to obtain the bound:
 \begin{equation}
 \label{boundeig}
\exp(-kJ(b-a)) \le \lambda  \le  \exp(kJ(b-a)),
\end{equation},
where $k$ is a bound on the size of the interval.
 
 The gap $\gamma$ of the spectrum can be bounded from below by the $\inf$ over the $\sup$ of the entries of the matrix (see e.~g. \cite{CCO}) . This gives
 \begin{equation}
 \label{gap}
 \gamma \ge \frac { \exp(-2Jk(b-a))} 
 {\prod_{j=0}^{k-1} \sinh (h_j(b-a)} 
 \ge \frac {\exp(-2Jk(b-a))} {(\sinh (C_1 (b-a)))^k }
\end{equation} 
uniformy.

\section{Main result\label{Result}}
The main result is the following theorem:

\begin{theorem} \label{theorem}
Given a constant $C_1>0$  such that  $h_x \ge C_1$ for $x$ belonging to some $A \subset \mathbb{Z}$
and the sizes of the clusters of $A$ are bounded by  some  constant $K$, then there is  $C_2 > 0$  such that if $h_x \ge C_2$ for
$x \in \mathbb{Z} \backslash A$, then $S( \rho^L_m)$ is uniformly bounded in $L$ and $m$.
\end{theorem}

\begin{proof}

As in \cite{CG} \cite{CG2} we perform a cluster expansion. 
In the present case we show that 
the constants $C_1$ and $C_2$ can be chosen in such a way that the conditions 
of of Koteck\'{y} and Preiss (\cite{KP}) are verified.

As we have previously remarked, the ground state can be obtained with any boundary conditions
on the upper and lower boundaries. It is convenient here to take a boundary conditions for which
the spin in $\Lambda \backslash A$ are independent and equal to $1$ and $-1$ with
probabilities $\frac {1} {2}$ and $\frac {1} {2}$, whereas the distribution of the spins
of a cluster of $A$  of size $m$ is proportional to $p$, the largest eigenvector of the transfer matrix corresponding to that cluster, independently of the other spins.

Let two positive constants $\delta_1$ and $\delta_2$ be given, with $\delta_2$ 
integer multiple of $\delta_1$.  The vertical lines above points $x \in \mathbb{Z} \backslash A$
are partitioned into intervals of size $\delta_1$ and those  above points $x \in A$ are partitioned into intervals
of size $\delta_2$.   The intervals are called respectively short and long intervals.
The  partitions are chosen to be compatible; that is possible as  $\delta_2$ 
is an integer multiple of $\delta_1$
$\mathcal{S}$ and $\mathcal{L}$ denote  respectively the sets of short and long intervals.

A graph is then built with vertices corresponding to intervals .Two vertices are connected by a horizontal edge
if the vertices are nearest neighbours and the intersection of their projections on the second coordinate
has non-empty interior. The system can be represented as a classical ``spin'' system on a graph, where
the  ``spins'' take value in a space of trajectories.

The trajectory of the spin-flip processes $(\sigma_x(t))$ 
for $x \in \Lambda \backslash A$  is partitioned into intervals of size $\delta_1$ and 
for $x \in  A$  is partitioned into intervals of size $\delta_2$.
 
A graph is then built with vertices corresponding to intervals .Two vertices are connected by a horizontal edge
if the vertices are nearest neighbours and the intersection of their projections on the second coordinate
has non-empty interior. 

We set $\mathcal{V}$ to be the set of vertical edges and $\mathcal{O}$ the set of horizontal edges.

The vertical
interaction  between two neighbouring vertices that imposes that the spin configurations of two vertically neighbour intervals are compatible. 

The interaction associated to an edge $b \in \mathcal{O}_1$ is given by 
\begin{equation} \label{horbond}
 W_b(\sigma, \sigma')=\int -J \sigma(t) \sigma'(t) \, \text{d}t
\end{equation}
where the integral is over the interval where the the second coordinates of the two intervals overlap and the paths are parametrized with the
second coordinate. For every edge $b \in  \mathcal{O}_1$ we make the decomposition
\begin{equation} \label{interaction}
\exp(W_b(\sigma, \sigma'))=1+ ( \exp(W_b(\sigma, \sigma')) - 1).
\end{equation}
The modulus of the second term on the r.~h.~s. of \ref{interaction} can be
bounded by
\begin{equation} \label{bound1}
| \exp(W_b(\sigma, \sigma')) - 1| \le J \delta_1 .
\end{equation}
The support corresponding to this term is the set of the coordinates of the spins
involved in \eqref{horbond}

The product  of $\exp(W_b(\sigma, \sigma')$ for $b \in \mathcal{O}_1$ , 
by means of the decomposition \eqref{interaction}, gives rise to a sum of terms
corresponding to subsets $B$ of $\mathcal{O}_1$.

 A second decomposition  is then 
performed corresponding to vertical edges.

 If there are two vertices in $\mathcal (S)$ with the same first coordinate  $v_1$, $v_2$ of $\Lambda_1 \cap A$
with $|x_1-x_2| \ge \delta_1$  
having the same first coordinate $i$, such that the open intermediate interval $(x_1, x_2)$ has empty intersection with $B$,
then if we integrate over the Poisson point process in the interval $[x_1, x_2]$ we get
$1 +\sigma_1 \sigma_2 \exp(-2 h_i |x_2 - x_1|)$, where $\sigma_1, \sigma_2$ are the
values of the spins at the extremes of the interval.  Since $i \in \mathbb{Z} \backslash A$, $h_i \ge C_2$, we have the bound
\begin{equation} \label{bound2}
 |\sigma_1 \sigma_2 \exp(-2 h_i |x_2 - x_1|| \le \exp(-C_2 |x_2 - x_1|).
\end{equation}

Consider a cluster $I$ of points in $A$  (since we are in one dimension it is just an interval of $l$ sites with $l \le K$, 
but this can be extended to more dimensions). We consider a maximal union $M=[a, b]$ of $k$ intervals of
length $\delta_2$ such that $B$ does not intersect $I \times M$.

$g_{b-a}( \sigma_a, \sigma_b)$ is the result
of the integral  over the Poisson point
processes in  $I \times M$, where $\sigma_a$ and $\sigma_b$ are the spin configurations at $I \times \{a \}$ and $I \times \{b\}$.
We can write 

\begin{equation} \label{operator1}
g_{b-a}( \sigma_a, \sigma_b) =
\int \exp \left( -J \int_a ^b\sum_{j=0}^{l-2}\sigma_j(t) \sigma_{j+1}(t)
\right)  \text{d}\bigotimes_{i=0}^{l-1} \mu_{h_i, [a,b]},
\end{equation}
 where the trajectories have prescribed initial and final values
 $\sigma_a$ and $\sigma_b$ and jumps prescribed by the Poisson point processes. The matrix $T$ with entries given
 $g( \sigma_a, \sigma_b) $ is positive and by Perron-Frobenius has a positive
 largest eigenvalue $\lambda$ and the remainder of the spectrum in a circle of strictly smaller radius. It is immediate to obtain the bound:
 \begin{equation}
 \label{boundeig1}
\exp(-KJ(b-a)) \le \lambda  \le  \exp(KJ(b-a)),
\end{equation},
where $K$ is the bound on the size of the interval.
 
 The gap $\gamma$ of the spectrum can be bounded from below by the $\inf$ over the $\sup$ of the entries of the matrix (see e.~g. \cite{CCO}) . This gives
 \begin{equation}
 \label{gap}
 \gamma \ge \frac { \exp(-2Jl(b-a))} 
 {\prod_{j=0}^{l-1} \sinh (h_j(b-a)} 
 \ge \frac {\exp(-2JK(b-a))} {(\sinh (C_1 (b-a)))^K }
\end{equation} 

We write $g_{b-a}( \sigma_a, \sigma_b) = p(\sigma_a)p_(\sigma_b)+(g_{b-a}( \sigma_a, \sigma_b)-p(\sigma_a)p(\sigma_b))$, where $p$ is 
is the eigenvector  corresponding to the largest (positive) eigenvalue $\lambda$ .
Multiplying by  $\lambda^{-1}$ we can let the largest eigenvalue become $1$. The multiplication corresponds 
to adding a constant to the hamiltonian. By Perron-Frobenius theorem with an estimate of the spectral gap with the ratio
between the inf and the sup of the entries of the transfer matrix, we get that  
\begin{equation} \label{equilibrum}
|g_{b-a}( \sigma_a, \sigma_b)-p(\sigma_a)p(\sigma_b)| \le c \cdot e^{-\gamma(b-a)},
\end{equation}
for some constant $c$ uniformly in the $h_x$'s $\ge C_1$ .

Finally consider a cluster $I$ in $A$ and an interval $M \in \mathcal{L}$ such that $I \times M$ intersect $B$. Then the integral
corresponding to $I \times M$ can be bounded by
\begin{equation} \label{bound4}
\exp \left(2J \delta_2 K \right)
    \end{equation}

After we have performed the previous decomposition, we obtain a sum 
of independent term. Each term is associated to a set of ``connected'' bonds:  two horizontal bonds are considered connected if their supports have non-empty intersection. Also in the case of vertical bonds we define their support as the set of the coordinates involved in their computation. The notion of connection of horizontal with vertical bonds or between vertical bonds is defined in terms of non-emptyness of the intersection of their supports.

So we can write the partition function as a sum of products:

\begin{equation} \label{dec}
Z = \sum \prod_i \zeta(\mathpzc{R}_i),
    \end{equation}

    where the $\mathpzc{R_i}$'s, called \emph{polymers}, correspond to disjoint terms of the previous decomposition and $\zeta(\mathpzc{R})$ , called \emph{activity}, is the value of the integral corresponding to the polymer $\mathpzc{R}$.

     Our main result follows from the convergence of \emph{cluster expansion} just as in \cite{CG} and this can be proved by verifying the basic inequality in the first theorem of \cite{KP}:
      or every polymer $\mathpzc{R}$

    \begin{equation} \label{KPinequality}
    \sum_{ \mathpzc{R'} : \mathpzc{R'} i \mathpzc{R}} 
    \exp \left(a(\mathpzc{R'})+d(\mathpzc{R'})\right) \le a( \mathpzc{R}),
    `\end{equation}

    where the sum is over all polymer $\mathpzc{R'}$ that intersect
    $\mathpzc{R}$ (i.~e. is incompatible with $\mathpzc{R}$.
     and $a(\mathpzc{R})$ and $d(\mathpzc{R}$) denote respectively the area and the diameter of the polymer $\mathpzc{R}$.

     After we have performed the previous decomposition, we obtain a sum 
of independent terms. Each term is associated to a set of ``connected'' bonds:  two horizontal bonds are considered connected if their supports have non-empty intersection. Also in the case of vertical bonds we define their support as the set of the coordinates involved in their computation. The notion of connection of horizontal with vertical bonds or between vertical bonds is defined in terms of non-emptyness of the intersection of their supports.

So we can write the partition function as a sum of products:

\begin{equation} \label{dec}
Z = \sum \prod_i \zeta(\mathpzc{R_i}),
    \end{equation}

    where the $\mathpzc{R_i}$'s, called \emph{polymers}, correspond to disjoint terms of the previous decomposition and $\zeta(\mathpzc{R})$ , called \emph{activity}, is the value of the integral corresponding to the polymer $\mathpzc{R}$.

     Our main result follows from the convergence of \emph{cluster expansion} just as in \cite{CG2} and this can be proved by verifying the basic inequality in the first theorem of \cite{KP}:
      for every polymer $\mathpzc{R}$

    \begin{equation} \label{KPinequality}
    \sum_{ \mathpzc{R'} : \mathpzc{R'} i \mathpzc{R}} 
    \exp \left(a(\mathpzc{R'})+d(\mathpzc{R'})\right) \le a( \mathpzc{R}),
    `\end{equation}

    where the sum is over all polymer $\mathpzc{R'}$ that intersect
    $\mathpzc{R}$ (i.~e. is incompatible with $\mathpzc{R}$.
     and $a(\mathpzc{R})$ and $d(\mathpzc{R}$) denote respectively the area and the diameter of the polymer $\mathpzc{R}$.

After we have performed the previous decomposition, we obtain a sum 
of independent terms. Each term is associated to a set of ``connected'' bonds:  two horizontal bonds are considered connected if the intersection of their supports has non-empty interior or if the upper and lower boundaries have non-empty (one-dimensional) interior, in the other cases two bonds are said to be connected if their supports have non.empty intersection.
We can bound the activity 
$\zeta( \mathpzc{R} )$ of a polymer $\mathpzc{R}$ with a product of factors:
\begin{equation} \label{bound4}
\zeta( \mathpzc{R} )  \le  \, \Pi_1 \, \Pi_2  \,\Pi_3 \,\Pi_4,
\end{equation}

where

\begin{itemize}
    \item $\Pi_1$ is the bound for horizontal bonds;
    \item $\Pi_2$ is the bound for vertical bonds that are considered as sequences of short intervals of length $\delta_1$ ;
    \item $\Pi_3$ is the bound for vertical bonds that are considered as sequences of long intervals of length $\delta_2$;
    \item $\Pi_4$ is the bound corresponding to long intervals where the transfer matrix cannot be used because they are connected to $B$.
\end{itemize}

Let  $N_1$, $N_2$, $N_3$, $N_4$ be the number of bonds involved in $\Pi_1$, $\Pi_2$, $\Pi_3$, $\Pi_4$ respectively.

In order to get the estimate that we need we have to fix the constants $C_2$, $\delta_1$, $\delta_2$. First we choose $\delta_2$,
that appears in the estimate \eqref{equilibrum}, the length of long intervals sufficiently large.  Then the constant $\delta_1$, the length of short intervals, that appears in the estimate \eqref{bound1}, is fixed sufficiently small as an integer submultiple of $\delta_2$. Finally the constant $C_2$, the lower bound of the transverse field in $\mathbb{Z} \backslash A$, that appears in the estimate \eqref{bound2}, is chosen sufficiently large.

By the previous estimates we have that for every $\theta > 0$ the constants $C_2$, $\delta_1$, $\delta_2$ can be chosen in such a way that
\begin{equation} \label{bound5}
\Pi_i \le \theta^{N_i} \quad \text{for} \quad i=1, 2, 3,
\end{equation}
whereas by \eqref{bound4} we have

\begin{equation} \label{bound6}
\Pi_4 \le \exp \left(2J k \delta_2 N_4 \right)
\end{equation}
We remark that 
\begin{equation} \label{bound0}
 N_4 \le 2 N_1;   
\end{equation}
indeed long intervals  that appear in in $\Pi_4$ must be connected at least to a horizontal bond in $\Pi_1$ and a horizontal bond in $\Pi_1$ can be connected to at most two long intervals in $\Pi_4$.
We also observe that for a polymer $\mathpzc{R}$ 
\begin{equation} \label{bound9}
a(\mathpzc{R}) +d(\mathpzc{R}) \le c_1(N_1+N_2+N_3+N_4),
\end{equation}
where $c_1$ is a suitable constant.

We can associate to a given polymer a connected graph with the vertices corresponding to the components and the bonds corresponding to the relation of connection. One can choose a tree contained in the graph. The tree identifies uniquely the corresponding graph given the connections of its components. We remark that the tree has bounded degree, since the number of components connected to a given component is bounded.

In order to prove that the inequality \eqref{KPinequality} is verified, we consider the sum of the activities of the polymers that intersect a given site.

Given the tree structure with bounded degree of the polymer, we can  take $\eta$ sufficiently small so that
\begin{equation} \label{bound7}
    \sum \eta ^{N_1+N_2+N_3+ N_4}  \le 1,
\end{equation}
where the sum rages over all polymers that intersect a given site. $\eta$ can be chosen independently from the site.

By the estimates \eqref{bound5} end \eqref{bound6} we have that given any $\theta>0$ for suitable choice of $C_2$, $\delta_1$, $\delta_2$ we have that 

\begin{align} 
\sum  & \zeta (\mathpzc{R}) \exp \left(a(\mathpzc{R}) + l(\mathpzc{R})\right) \le \\
&\sum \theta^{N_1+N_2+N_3} \exp( 2 c_1 (N_1+N_2+N_3+N_4)) \exp \left(2J k \delta_2 N_4 \right) \le  \\
& \theta^{N_1/2+N_2+N_3+N_4/2} \exp( 2 c_1 (N_1+N_2+N_3+N_4)) \exp \left(2J k \delta_2 N_4 \right) \le   \label{bound8} \\
& \sum \eta ^{N_1+N_2+N_3+ N_4}  \le 1 , \label{bound9}
\end{align}  
for $\theta$ sufficiently small, where $\eta$ is the constant introduced in \eqref{bound7} and in \eqref{bound8} we have exploited \eqref{bound0}.
The inequality \eqref{KPinequality} is obtained by adding \eqref{bound9} over the sites of some polymer. 

The convergence of cluster expansion with the exponential decay of the activity of polymers are the ingredients of the proof of the uniform bound on the entanglement entropy as it is carried out in \cite{CG2} (based on \cite{GOS}). 
The same argument works in the present case of non-homogeneous transverse field.

 \end{proof}

 \begin{remark} The theorem applies in particular to the case of periodic transverse field. It states for example that if we fix in an arbitrary way the values of the field in all points in the period except one, the entanglement entropy is uniformly bounded provided that the value at the selected point is sufficiently large.
 \end{remark}
\begin{remark}
The proof of convergence of cluster expansion extends immediately to any dimension and also, with suitable readjustments, to the Ising model or other statistical mechanical models on the discrete lattice $\mathbb{Z}^{d+1}$, where the interaction can vary in the first $d$ coordinates and is constant in the last one.
\end{remark}

\section*{Acknowledgement} 
M. Campanino is
member of G.~N.~A.~M.~P. A.. 

\textit{To Abel Klein with friendship.}

\end{document}